\documentclass[letterpaper,11pt]{article}

\usepackage{fullpage}

\usepackage{amsmath}
\usepackage{amsfonts}
\usepackage{amsthm}

\usepackage{authblk}

\usepackage{subfig}

\usepackage{amsthm,amsmath,amssymb,amsxtra}
\usepackage[linesnumbered,ruled]{algorithm2e}

\usepackage{graphicx}
\usepackage{caption}
\usepackage{subfig}
\usepackage{float}
\usepackage{epsfig}

\usepackage{tikz}
\usetikzlibrary{calc}

\newtheorem{theorem}{Theorem}
\newtheorem{lemma}[theorem]{Lemma}

\newtheorem{corollary}[theorem]{Corollary}

\newcommand{\opt}{\textrm{opt}}
\newcommand{\AdSch}{{\sc AdSch}}
\newcommand{\AdSub}{{\sc AdSub}}
\newcommand{\WS}{\sc WS}

\newcommand{\mathsc}[1]{{\normalfont\textsc{#1}}}

\title{Optimal Strategies for Weighted Ray Search}

\author[1]{Spyros Angelopoulos}
\author[2]{Konstantinos Panagiotou}
\affil[1]{CNRS and Sorbonne   Universit\'e,  Laboratoire   d’Informatique  de  Paris   6, Paris, France. {\tt spyros.angelopoulos@lip6.fr}}
\affil[2]{LMU Munich, Munich, Germany. {\tt kpanagio@lmath.lmu.de}}

\date{}

\begin{document}

\maketitle

\begin{abstract}
We introduce and study the general setting of {\em weighted search} in which a number of targets, each with a certain weight, are hidden in a star-like environment that consists of $m$ infinite, concurrent rays, with a common origin. A mobile searcher, initially located at the origin,  explores this environment so as to locate a set of targets whose aggregate weight is at least a given value $W$. The cost of the search strategy is defined as the total distance traversed by the searcher, and its performance is measured by the worst-case ratio of the cost incurred by the searcher over the cost of an on optimal, offline strategy with complete access to the instance. This is the first study of a setting that generalizes several problems in search theory: the problem in which only a single target is sought, as well as the problem in which all targets have unit weights.

We present and analyze a search strategy of near-optimal performance for the problem at hand. 
We observe that the classical approaches that rely on geometrically increasing search depths perform  rather poorly in the context of weighted search. We bypass this problem by using a 
strategy that modifies the search depths adaptively, depending on the number of targets located up to the current point in time.
\end{abstract}



\section{Introduction}

In this paper we introduce and study the following search problem. We are given a star-like environment that consists of $m$ concurrent {\em rays}
of infinite length, with a common origin~$O$. For each ray $i \in \{0,\ldots,m-1\}$, there is a {\em target} of weight $w_i \geq 0$ that is hidden at some distance $d_i\geq 0$ from~$O$. Note that the setting allows cases such as $d_i=\infty$ (i.e., there is no target hidden on ray $i$) 
and $w_i=0$ (i.e., the target has no weight). A mobile searcher is initially located at $O$, and its objective is to locate a subset of 
targets whose aggregate weight is at least a specified value $W$. The searches knows the number of rays $m$, however, it 
has no further knowledge about the instance, namely the weights of the targets and their distances from~$O$. 

A {\em search strategy} $\Sigma$ for this problem is specified by determining, at every point in time that the searcher is at the origin, 
a ray $r$ and a depth $\ell_r$, such that the searcher will next traverse ray $r$ to depth at most $\ell_r$ and then will return back to $O$. If during this search a target was found, and as long as the weight accrued is less than $W$, the searcher will immediately return to $O$ and will never again traverse ray 
$r$. In general, we assume that the search strategies have {\em memory};  the pairs $(r,\ell_r)$ may depend, among other things, on targets already discovered by $\Sigma$ and the already explored depths.

As it is quite common in such problems, we will evaluate the performance of a search strategy $\Sigma$ 
by means of the well-established {\em competitive ratio}, which can be traced to early work in~\cite{beck:yet.more} 
in the context of the linear search problem. Informally, the competitive ratio compares the performance of a search strategy
that is oblivious of the instance (that is, the exact position of the targets and their weights) to an omniscient searcher that has full information of the instance. To formalize this concept, given a number of rays $m \ge 2$, let us
denote by ${\cal I}_m$ the set of all instances to the search problem, namely
$${\cal I}_m=\left\{\big(W, (d_i,w_i)_{0 \leq i \leq m-1\}}\big):
W \geq 0, w_i\geq 0, d_i \geq 1, \text{ and } \sum_{0\le i\le m-1} w_i\geq W\right\}.$$
Note that we make two standard assumptions in the field, namely that 
all targets are at least a unit distance from $O$ -- otherwise, no strategy can have a bounded competitive ratio -- and that there is a feasible solution to the instance.
Given $I \in {\cal I}_m$   the {\em cost of $\Sigma$ on $I$}, denoted by $c(\Sigma,I)$ is defined as the total distance traversed by the searcher until the first time it discovered targets of aggregate weight at least $W$. We also denote by $\opt(I)$ the {\em optimal cost of $I$}, namely the cost of an ideal strategy
that has complete knowledge of the instance $I$ (i.e., the positions of the targets and their weights, as well as $W$). 
A strategy $\Sigma$ is called {\em $\rho$-competitive}  for some $\rho\geq 1$, if $c(\Sigma,I) \leq \rho \cdot \opt(I)$, for all $I \in {\cal I}_m$.
The competitive ratio of $\Sigma$ is then defined as 
\begin{equation}
\rho(\Sigma)=\sup_{I \in {\cal I}_m} \frac{c(\Sigma,I)} {\opt(I)},
\label{eq:competitive.ratio}
\end{equation}
and reflects the overhead of $\Sigma$ due to lack of information. A strategy of minimum competitive ratio is called {\em optimal}.
Figure~\ref{fig:weighted} illustrates an example of an instance. 
\begin{figure}[htb!]
    \centering
    \includegraphics[height=5.5cm, width=9cm]{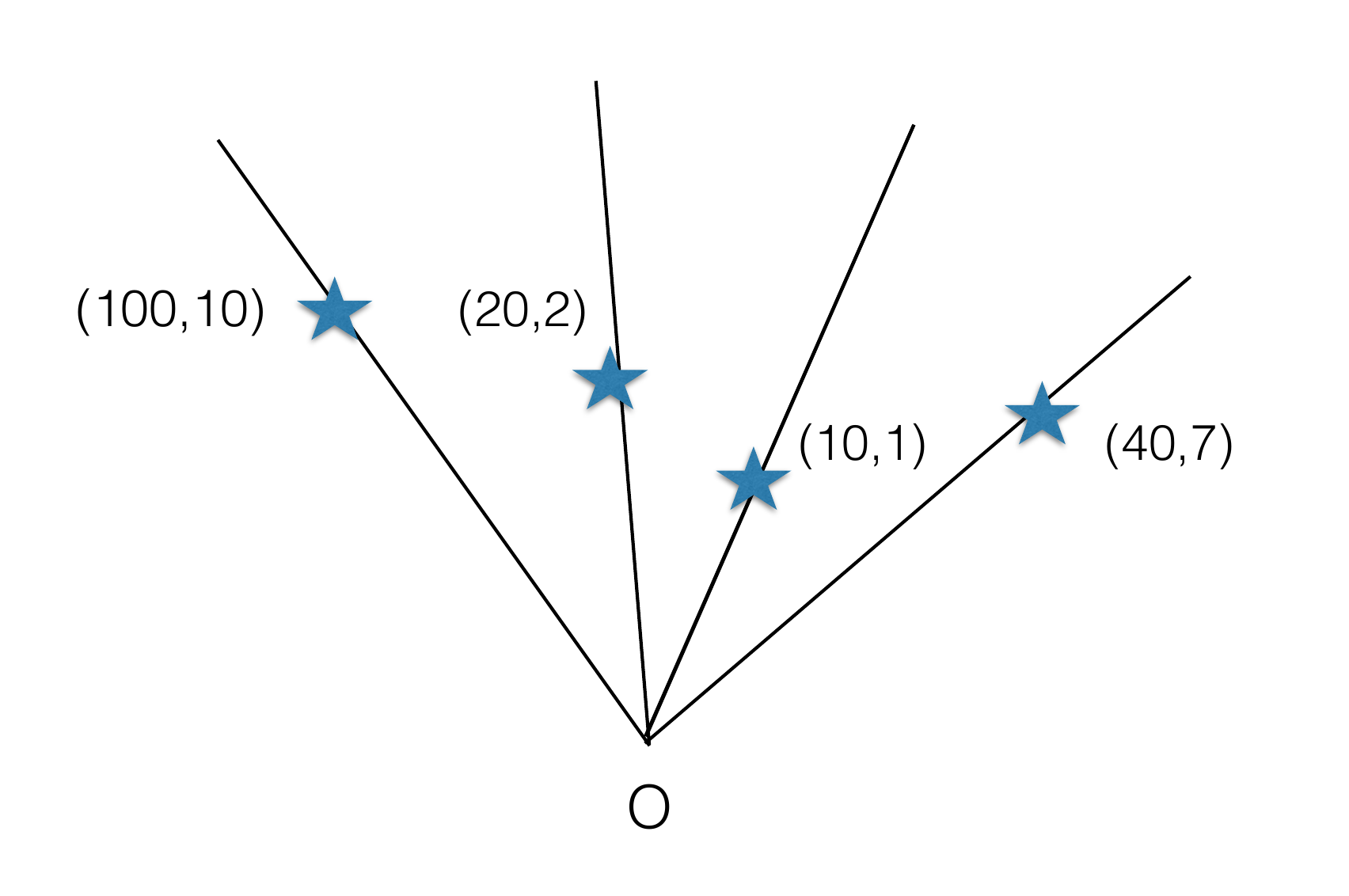}
  \caption{Illustration of the weighted star search problem, for an instance $I$ with $m=4$ rays. The target at ray $i$ is represented by the pair 
  $(d_i,w_i)$, where $d_i$ denotes its distance from $O$ and $w_i$ its weight. Assume an indexing of rays as $0\ldots 3$, from left to right, and a value of $W$ equal to 10. Then an optimal strategy will either search the ray 0 at a cost of 100 or the rays 1,2,3 (in this order), which also yields the optimal cost of $2\cdot 20 + 2\cdot 10 + 40 = $ 100.}
  \label{fig:weighted}
\end{figure}


Our setting is motivated by many factors. First, it generalizes several well-studied problems, most notably the 
(unweighted) {\em ray search} or {\em star search} problem in~\cite{gal:minimax}, 
in which a single target is present on some ray.
Note that star search is a generalization of {\em linear search}, known informally as the cow-path 
problem, originally studied in~\cite{bellman} and~\cite{beck:ls}. In addition, the weighted setting subsumes the multi-target search of~\cite{multi-target}, in which each target has unit weight, and the objective is to locate~$t$ targets, for some~$t \in {\mathbb  N}^+$ that is known to the strategy.

Furthermore, weighted search provides a useful abstraction of the setting in which one has to allocate resources to different tasks without advanced knowledge of the degree to which each task will prove itself fruitful. This is a fundamental aspect of many decision-making processes. For instance, a researcher may want to determine how to allocate her time among $m$ different projects, without knowing in advance how useful the outcome of each project will be towards a combined publication of the potential results. Here, the projects can be modeled as weighted targets, with the distance to a target representing the time needed to the successful completion of the corresponding project. For a different example, drilling for oil in multiple locations can also be modeled as a star-search problem; see the discussion in~\cite{oil}. 

In addition, we emphasize that ray searching can often model settings that go far beyond the search paradigm. 
Notable examples include: the design of {\em interruptible} algorithms, i.e., algorithms that return acceptable solutions even if interrupted during their execution (see~\cite{steins,spyros:ijcai15});
the synthesis of {\em hybrid} algorithms based on a suite of heuristics shown in~\cite{hybrid}; and database query optimization (in particular, pipeline filter ordering as studied in~\cite{Condon:2009:ADA:1497290.1497300}, which has an equivalent search-related formulation over a star graph as discussed in~\cite{stacs-expanding}). 
Thus, solutions to the weighted search problem can provide a framework for solving weighted variants of the above settings. For instance, one may define hybrid algorithms in which each heuristic execution is assigned a score related to the quality of the returned solution, or a system of interleaved executions of interruptible algorithms in which each algorithm is assigned a priority based on its importance.

For the above reasons, star search can be seen as a type of {\em scheduling} or {\em sequencing} problem, in that we seek a strategy for allocating resources (available search time) to different tasks (the set of rays). See, e.g.,~\cite{steins} for an explicit interpretation of star search as a scheduling problem. In this sense, our setting is motivated by considerations that are natural in scheduling theory, namely the transition from the simple, unweighted variants to more complex and inclusive weighted ones.

\subsection{Contribution}

In this work we present and analyze a search strategy called {\sc AdaptiveSearch} (or \AdSch \ for brevity) 
for weighted ray search that attains a (near-)optimal competitive ratio. Define the function $\phi$ as
\[
\phi(x)=1+2\left(1+x\right)\left(1+\frac{1}{x}\right)^{x}, \textrm{ $x>0$}.
\]
Note that $\phi(x)=\Theta(x)$ as $x$ gets large. As we will see, function $\phi$ plays a central role in the analysis and reveals connections to  previous studies of the simpler variants of search problems; we will elaborate more on that in due course. Our first theorem provides a tight, worst-case analysis of the competitive ratio of weighted search. 
\begin{theorem} 
The competitive ratio of \AdSch \ is $\phi(m-1)$, and this is the optimal competitive ratio for weighted search.
\label{thm:no.constraint}
\end{theorem}
Some remarks are in place. Let us consider the special case of the search problem we study, namely the classic problem of searching on $m$ rays. Here, there is precisely one target that hides in some ray unknown to the searcher. Using our notation, this problem corresponds to the case  $d_i = \infty$ for all rays $0 \le i \le m-1 $ except for exactly one ray $i^*$, and $W = w_{i^*}$. The competitive ratio of this problem was determined by~\cite{gal:minimax}, who showed that it is equal to $\phi(m-1)$. It is then immediate that the competitive ratio of weigthed search is at least $\phi(m-1)$; the main contribution of Theorem~\ref{thm:no.constraint} is to establish that this is also an upper bound. In other words, our first result says that in the worst case, weighted search is \emph{at most} as difficult as the conceptually simpler problem of searching one target in $m$ rays. 

Even though Theorem~\ref{thm:no.constraint} is tight for some instances, a moment of thought reveals that it actually provides a rather too pessimistic worst-case guarantee. In order to understand this better, let us consider another special case of weighted search, namely {\em multi-target} search in $m$ rays. In this variant, introduced and studied in~\cite{multi-target}, the searcher must locate exactly $1\le t \le m-1$ targets that are placed on rays again unknown to the searcher; in our notation this corresponds to $d_i = \infty$, for all except of $t$ rays $0\le i_1^*< \dots < i^*_t\le m-1$ with weights $w_{i^*_1} = \dots = w_{i^*_t} = W/t$. In~\cite{multi-target} it was shown that the competitive ratio of this problem is $\phi(m-t)$, which can be rephrased as ``searching for $t$ targets on $m$ rays is as difficult as searching for one target in $m-t+1$ rays''. Note that the larger~$t$ is, the better the competitive ratio becomes.

The results in~\cite{multi-target} leave open the possibility that more refined guarantees than $\phi(m-1)$ may be attainable for weighted search; in particular, one expects that if many targets have to be located, then the competitive ratio should decrease. 
In order to make this precise, we follow a {\em parameterized} approach, in which we study the performance of a strategy 
not only in terms of $m$, but also in terms of some natural parameters that reflect the ``hardness'' of the instance. 
More specifically
given  $I= \{(W, (d_i,w_i)_{0 \leq i \leq m-1\}})\}$, define $W_I=W$ and $w_{\max,I}=\max_{0\leq i \leq m-1} w_i$. 
Then any strategy, including the offline optimal one that knowns $I$, must locate at least $\lceil W_I /w_{\max,I} \rceil$ targets. The question then 
is: what is the competitiveness of weighted search, in terms of both $m$ and the above bound on the number of targets that must be found?
Our next result gives an near-tight answer, and generalizes the known results on multi-target search.
\begin{theorem}
Let $I \in {\cal I}_m$ be such that $w_{\max,I}>0$. If $\lceil W_I /w_{\max,I} \rceil < m$, then 
\[
c(\normalfont\mathsc{AdSch},I) \leq \phi \left( m-\lceil {W_I}/{w_{\max,I}} \rceil \right) \cdot \opt(I),
\]
and $c(\normalfont\mathsc{AdSch},I) \leq (3+2e) opt(I)$ otherwise. Moreover, for every strategy $\Sigma$, there exists 
$I \in {\cal I}_m$ with $\frac{{W_I}}{{w_{\max,I}}}<m$ such that 
\[
c(\Sigma,I) \geq \phi \left( m-\lceil {W_I}/{w_{\max,I}} \rceil \right) \cdot \opt(I),
\]
and there exists $I \in {\cal I}_m$ with ${{W_I}}/{{w_{\max,I}}}=m$
such that $c(\Sigma,I) \geq 2 \opt(I)$. 
\label{thm:weights}
\end{theorem} 
While this result is essentially tight with respect to the parameter ${W_I}/{w_{\max,I}}$, we can, obtain an even stronger performance guarantee. As we show in the next theorem, we can surprisingly relate the competitiveness of weighted search to the {\em maximum} number of targets that an ideal solution of optimal
cost must locate; in contrast, in Theorem~\ref{thm:weights}, the competitive ratio is  expressed in terms of the {\em minimum} number of targets in any ideal solution (more precisely: a lower bound on this number).
For a general instance $I \in {\cal I}_m$, let  $s_I$ denote the maximum value of this parameter, see Section~\ref{sec:prelim} for a formal definition. 
To illustrate this concept, consider the instance of Figure~\ref{fig:weighted}. There are two optimal solutions: 
one locating the target $(100,10)$, and one that locates all other targets. Thus, for this instance, we have that $s_I=3$.
\begin{theorem}
Let $I \in {\cal I}_m$. Then 
\[
c(\normalfont\mathsc{AdSch},I)\leq \phi(m-s_I) \cdot \opt(I), \text{ if } s_I<m,
\]
and  $c(\normalfont\mathsc{AdSch},I) \leq (3+2\textrm{e}) \opt(I)$ otherwise. 
Moreover, for every strategy $\Sigma$ and for every $s\in\{1, \ldots ,m-1\}$, there exists $I \in {\cal I}_m$ such that
\[
c(\Sigma,I)\geq \phi(m-s_I) \cdot \opt(I),
\] 
and there exists $I$ with $s_I=m$ such that $c(\Sigma,I)\geq (2-\epsilon) \opt(I)$, for arbitrarilly small $\epsilon>0$.
\label{thm:subsets}
\end{theorem}
The extreme cases $W_I/w_{\max,I}=m$ and $s_I=m$ in the statements of Theorems~\ref{thm:weights} and~\ref{thm:subsets} respectively, 
describe the outlier cases in which the optimal strategy must locate all $m$ targets. In this case, a small additive gap of at most
$(1+2\textrm{e}) \cdot \opt(I)$ remains.


It is worth noting that a conceptually related parameterized analysis has been applied to standard ray searching for a single target. 
For example, even though the strategy of~\cite{gal:minimax} achieves the best competitive ratio, when 
parameterizing with respect to upper and/or lower bounds on the distance of the target, more refined results can be 
obtained (as in~\cite{ultimate,revisiting:esa,jaillet:online}). 


\subsection{Techniques and analysis}   
In this work we rely on a strategy that searches the rays in a {\em cyclic}, i.e., round-robin manner; this is motivated by the fact that, for
the special case of the $m$-ray search problem,~\cite{Gal80} showed that any optimal search strategy must be cyclic. 
There are two main challenges that one needs to address. The first challenge pertains to the design of the strategy itself, 
and more precisely in determining the appropriate sequence of search lengths. 
Unlike most previous work in which for an optimal strategy it suffices to increase the search lengths by the same factor at each step, in our setting 
we show later that such simple rules are very inefficient. We thus introduce an {\em adaptive} strategy in which the search lengths increase depending on the total number of targets found by the end of each step. The second challenge lies in analyzing this strategy, since the setting is substantially more complex than the unweighted one. 
To this end, we relate the competitiveness of the weighted search problem to that of a related 
problem which we term {\em subset search}.
The objective in the latter problem is, informally, to locate a certain subset of the targets, without advance knowledge of the specific subset. 
For this problem, which is of independent interest, we show strategies that achieve the same competitiveness as strategies that 
search for {\em any} subset of targets that has the same cardinality as the subset that is sought. This, perhaps surprising, result 
establishes a strong upper bound on the competitiveness of the problem and is the main technical challenge. 

To arrive at this result, we first prove essentially tight bounds on the optimal cost, by identifying appropriate parameters related to the search, and by relaying the cost of our strategy to precisely those parameters. We then apply an inductive argument on the number of targets discovered by the strategy in the last $m$ rays that have been explored by the searcher. Some technical difficulties arise from the fact that, unlike the single-target variant, there is no simple expression that describes the optimal cost, and that the cost of the searcher involves several parameters in a trade-off relation. 

\subsection{Further related work}
\label{subsec:related.work}


The $m$-ray search problem, also known as \emph{star search}, is a generalization of the {\em linear search} problem and has been studied in many variants and settings. It is known that any optimal search strategy must be defined as a geometric sequence of search lengths (\cite{gal:general, Gal:exponential}). We refer the reader to  Chapters 8 and 9 in the book~\cite{searchgames} for results related to search games on a star. 
Optimal \emph{randomized} strategies for linear search ($m=2$) were given in~\cite{ray:2randomized}, and some results concerning the value of the star search games (assuming general mixed strategies) were shown in~\cite{hybrid,schuierer:randomized,Kella:star.search}. The setting in which multiple searchers explore the star was studied in~\cite{alex:robots}, whereas 
multi-target searching was studied in~\cite{hyperbolic,oil,multi-target}. Other variants include:
searching with turn cost (\cite{demaine:turn, Angelopoulos2017}); searching with an upper bound on the target distance (\cite{HIKL99:fixed.distance,revisiting:esa}); the variant in which some probabilistic information on target placement is known (\cite{jaillet:online, informed.cows}); 
and expanding search (\cite{stacs-expanding} and~\cite{ejor:competitive.ratio}).

\section{Notation \& $m$-ray search}
\label{sec:prelim}

We begin with some useful definitions and facts concerning the problem we study. Given an instance $I \in {\cal I}_m$ of 
{\sc WeightedSearch} ({\WS}), we denote by $T_I = \{0, \dots, m-1\}$ the set of the $m$ targets of $I$. 
Given $S\subseteq T_I$, we define the {\em total weight of $S$}, denoted by $w_S$ as 
$\sum_{i\in S} w_i$, and the {\em optimal search cost of $S$}, denoted by $d_S$ as 
\[ 
d_S=2\sum_{i\in S}d_i-\max_{i\in S}d_i.
\] 
The latter is the minimum cost required by the mobile searcher initially located at the origin to visit all targets in $S$ assuming full information of $I$ (note that the ray with the most distant
target among targets in $S$ is visited last, and that the searcher does not need to return to the origin after having located all targets).  With this notation the optimal cost of the instance $I$ is 
\[
\opt(I)=\min_{S \subseteq T_I} \{d_S:w_S\geq W\}.
\]
Moreover, we denote by $s_I$ the largest number of targets in an optimal solution, that is,
\[ s_I=\max \big\{|S| : S\subseteq T_I, w_S \geq W, \text{and } d_S= \opt(I)\big\}.
\] 
The following quantities will appear several times in the forthcoming analysis. Define
\begin{equation}
\label{eq:defb}
	b_x = 1 + \frac1x
	\quad \text{and} \quad
	b_{m, t} = b_{m-t} = 1 + \frac{1}{m-t}.
\end{equation}
Note that for the function $\phi$ defined in the introduction we obtain the useful relation
\[
\phi(x)=1+2b_{x}^{x+1}/(b_{x}-1).
\]
Let us now turn our attention to a special case of weighted search considered in the introduction, namely the classical $m$-ray search problem in which
 there is exactly one target to be found. 
For this problem, it was shown in~\cite{gal:minimax}  that there exists an optimal cyclic strategy with geometrically increasing lengths. More specifically, in the $i$-th step, the strategy 
searches ray $i \pmod m$ up to length $b_{m,1}^i$. As already mentioned, the (optimal) competitive ratio equals
\[
1+2\frac{b_{m,1}^m}{b_{m,1}-1}=\phi(m-1). 
\]
Concerning the problem in which there are $m$ unweighted targets, one per ray, and the searcher seeks to locate any subset of exactly $t$ targets, 
it was shown in~\cite{multi-target} that there is an optimal cyclic strategy, which in the $i$-th step searches the corresponding ray up to length 
$b_{m-t+1,1}^i$. The optimal competitive ratio is shown to be equal to 
\[
1+2\frac{b_{m-t+1,1}^{m-t+1}}{b_{m-t+1,1}-1}=\phi(m-t),
\] 
namely the same as the 
competitive ratio of searching for a single target in a star consisting of $m-t+1$ rays.  Note that in the unweighted setting, the searcher benefits from knowing the number $t$ of targets that are sought. This allows~\cite{multi-target} to give an non-adaptive optimal strategy, i.e., one in which in the $i$-th step the searcher goes up to distance $b^i$ from $O$, for some appropriate $b$. As we will argue in the next section, this type of strategy cannot be efficient for weighted search, because the ``right'' number of targets that needs to be located is a parameter that is unknown to the searcher.

\section{Strategies and the Overall Approach}
\label{sec:algorithm}
In this section we present the strategy for weighted star search, and the main approach to the analysis, namely the relation between the competitiveness
of {\WS} and {\sc SubsetSearch} ({\SS}).  We formally define the latter problem as follows. The instance to the problem 
consists of $m$ unweighted targets (one per ray) with the target at ray $i$ being at distance $d_i$ from the origin of the star, as well as a subset $S \subseteq T_I$ of the targets. The distances, as well as the subset $S$ are not known to the searcher. The search terminates when all targets in 
$S$ have been discovered; we can assume the presence of an oracle that announces this event to the searcher and thus prompts the termination. 
The cost of the search is defined as the total cost incurred at termination, whereas the cost of the optimal solution to the instance 
is the cost for locating all targets in $S$ assuming full information of the instance, i.e., equal to $d_S$.
The following lemma establishes a useful association between the two problems.
\begin{lemma}
Suppose there is a strategy $\Sigma_s$ for {\SS} such that for any instance $I_s=(S,d_0, \ldots d_{m-1})$ we have that
$c(\Sigma_s,I_s) \leq \rho(|S|,m)\cdot d_S$, where $\rho(|S|,m)$ is a function of $|S|,m$. Then there is a strategy $\Sigma_w$ for {\WS}
such that for any instance $I_w$ of this problem, $c(\Sigma_w,I_w) \leq \rho(s_I,m) \cdot \opt(I_w)$. 
\label{thm:reduction}
\end{lemma}

\medskip
\begin{proof}{Proof.}
Let $I_w=\{W,(d_0,w_0), \ldots ,(d_{m-1},w_{m-1})\}$ denote an instance of {\WS}.
Consider the instance $I_s$ for {\SS} in which
$m$ targets are at distances $d_0, \ldots d_{m-1}$, and $S$ is defined to be any subset of targets such that 
$d_S = \opt(I_w)$ and $|S|=s_I$. 
From the definition of $s_I$, such a set exists. 
Define  $\Sigma_w$ for {\WS} as follows: $\Sigma_w$ executes $\Sigma_s$
until an aggregate target weight of at least $W$ has been located. Since 
$w_S\geq W$, it follows that $c(\Sigma_w,I_w)\leq c(\Sigma_s,I_S)$. Furthermore, since $\Sigma_S$ has competitive ratio $\rho(s_I,m)$, we have that 
$c(\Sigma_s,I_S)\leq \rho(s_I,m) d_S$. 
Thus, $c(\Sigma_w,I_w)\leq \rho(s_I,m) \opt(I_w)$.
\end{proof}
\bigskip

Lemma~\ref{thm:reduction} demonstrates that the competitiveness of {\SS} is directly related to that of  {\WS}. 
We thus focus on {\SS}; in particular, we will analyze a strategy which we call {\sc AdaptiveSubset} ({\AdSub} for brevity),
and which is described by the statement of Algorithm~\ref{alg:exponential.deterministic}. The strategy explores rays in cyclic order 
(line 13), and keeps track of the found targets, denoted by $f$ (line 11). 
On each iteration, the strategy will search a ray to a length equal to $b_{m,f}$ times the length of the last ray exploration that did not 
reveal a target (line 8), until a new target is discovered. 
Once a target is found on a ray, the ray is not considered again; namely, the remaining rays are relabeled so as to
maintain a cyclic order (line 12). 

\SetNlSkip{0.8em}
\IncMargin{1.5em}
\begin{algorithm}[htb!]
\DontPrintSemicolon
\caption{Strategy \underline{{\AdSub}} for {\SS}}
\label{alg:exponential.deterministic}

{\bf Input:} $m$ rays labeled $\{0,\dots, m-1\}$, subset oracle ${\cal O}$\;

$f \leftarrow 1$, $r \leftarrow 0$, $D \leftarrow 1$

\Repeat {all targets, according to ${\cal O}$,  were found}
{
\Repeat{a target was found}
{
explore the $r$th ray up to distance $D\cdot b_{m,f}$ or until a target is found\;
\If{no target was found}
{
$r \leftarrow r + 1 ~ (\bmod~m-f+1)$ \;
$D \leftarrow D \cdot b_{m,f}$ \;
}
}
$f \leftarrow f + 1$ \;
remove the $r$th ray and relabel the rays canonically from $0\dots m-f$ \;
$r \leftarrow r  ~ (\bmod~m-f+1)$ \;
}
\end{algorithm}
\DecMargin{1.5em}

Before we proceed with the analysis, it is instructive to point out that the need for an adaptive strategy that modifies the search lengths as a function of the number of targets that have been found seems to be, in a sense, unavoidable. To see this, consider first a 
cyclic strategy that uses geometrically increasing lengths with a fixed base that may depend only on $m$. Consider also an instance in which $|S|-1$ targets from $S$ are very close to the origin, then the competitive ratio of this strategy is (more or less) at least the competitive ratio of a strategy that searches one target in $m-|S|+1$ rays; crucially, this ratio can only be achieved if the base equals $b_{m-|S|+1,1}$; see~Section~\ref{sec:prelim}.  Since $|S|$ is not known in advance, the wrong choice of a base can have a detrimental effect on performance.

The above example demonstrates the need for an adaptive  strategy that modifies the search lengths as a function of the number of targets that have been found. However, it is not entirely clear how to adapt the search lengths. To demonstrate this,
let us consider an obvious candidate, which also turns out to be inefficient. Consider a cyclic, geometric strategy that changes the base of search lengths once a target is found; more precisely, a strategy that on the $i$-th step searches the corresponding ray up to depth $b_{m,t}^i$, where $t-1$ represents the number of targets found at the beginning of the step. This strategy has an unbounded competitive ratio. To see this, suppose that $|S|=1$, and that in iteration   $i$, the strategy finds on ray $r$ a target that is not in $S$. Suppose also that the unique target in $S$ lies on ray $(r-1) \bmod m$ and  at distance $b_{m,1}^{i-1}+\epsilon$, for some $\epsilon>0$. The strategy will discover this target after having spent cost at least $2b_{m,2}^{i+m-2}$ (namely, the cost for searching ray $(r-2) \bmod m)$ on iteration $m+i-2$). It follows that the competitive ratio is at least $(b_{m,2}/b_{m,1})^i$, which is unbounded since $i$ can be arbitrarily large. 

By combining the insights from both previous examples, we see that not only we have to adapt the search lengths, but we also have to ensure a smooth transition when such a modification takes place. This is precisely accomplished in {\AdSub}; when a target is found, say in iteration $i$, the search depths in the following rounds do not jump `abruptly' from $b^{i+j}$ to some $\tilde{b}^{i+j}$, but instead to $b^i \cdot \tilde{b}^j$, $j \ge 0$.
Concerning weighted search, Lemma~\ref{thm:reduction} guides our choice of strategy. Namely, we
denote by {\AdSch} the strategy for {\WS} obtained by considering $\Sigma_s$ to be {\sc AdSub} 
in the statement and proof of Lemma~\ref{thm:reduction}. 

The following lemma, which bounds the competitive ratio of \AdSub, is the main technical result of this work. 
Its proof is given in Section~\ref{sec:analysis}. 
\begin{lemma}
Let $m \geq 2$. Given an instance $I$ of {\SS} in which we seek a set $S\subseteq \{0, \ldots ,m-1\}$, 
with $S \neq \emptyset$, we have 
\begin{equation}
c(\normalsize\mathsc{AdSub},I) \leq \phi(m-|S|) \cdot d_S, \ \text{ if }\ |S| \leq m-1,
\label{eq:upperBoundSubset}
\end{equation}
and $c(\normalsize\mathsc{AdSub},I) \leq (3+2e)\cdot d_S$, \ \text{ if } $|S|=m$.
\label{thm:SubsetSearchUpper}
\end{lemma}

Assuming Lemma~\ref{thm:SubsetSearchUpper}, we show first how to obtain Theorem~\ref{thm:no.constraint}. This establishes a tight, albeit worst-case 
bound on the competitive ratio. 

\medskip
\begin{proof}{Proof of Theorem~\ref{thm:no.constraint}.}
From Lemmas~\ref{thm:reduction} and~\ref{thm:SubsetSearchUpper}, for all $I \in {\cal I}_m$, 
\[
c(\normalsize\mathsc{AdSch},I) \leq \max\{3+2e,\phi(m-1)\} \cdot \opt(I),
\]
since $\phi$ is increasing. Moreover, $\phi(1)=9>3+2e$, hence the upper bound follows. 
This bound is tight for the instance $I$ in which there is only one target of weight $w$, and $W_I=w$
(i.e., standard star search for a single target in~\cite{gal:minimax}). 
\end{proof}
\bigskip

In order to prove Theorems~\ref{thm:weights} and~\ref{thm:subsets}, we will need the following lemma 
that addresses the extreme cases.
\begin{lemma}
Given an instance $I$ of {\WS} for which $s_I=m$ or $W/w_{\max,I}=m$, we have $c(\normalsize\mathsc{AdSch},I) \leq (3+2e) \opt(I)$.
Moreover, for every strategy $\Sigma$, there exists $I$ such that $c({\Sigma},I) \geq (2-\epsilon) \opt(I)$, 
for arbitrarily small $\epsilon$.
\label{thm:spastiko.case}
\end{lemma}

\medskip
\proof{Proof.}
The upper bound follows directly from the upper bounds of Lemmas~\ref{thm:reduction} and~\ref{thm:SubsetSearchUpper}.  
Remains thus to show the second part of the lemma.
To this end, consider any strategy $\Sigma$ that has to locate all $m$ targets (as stipulated by the assumption). Consider any snapshot
of the execution of the strategy at the moment the searcher returns to the origin, and let $l_i$ denote the depth at which ray $i$ has
been searched. Consider an instance $I$ in which $m$ targets of unit weights are placed such that there is a target at distance $l_i+\epsilon$,
for arbitrarily small $\epsilon$, and $W_I=m$. Then $c(\Sigma,I)\geq 2\sum_{i=0}^{m-1}l_i+\opt(I)$, and for sufficiently small $\epsilon$, we
have that $2\sum_{i=0}^{m-1} l_i \geq (1+\epsilon)\opt(I)$, which yields the result. 
\endproof
\bigskip
Having established bounds for the extreme cases, we can now proceed with the proofs of Theorems~\ref{thm:weights} and~\ref{thm:subsets}.
Recall that these theorems establish our main parametric results on the competitive ratio of weighted search.

\medskip
\begin{proof}{Proof of Theorem~\ref{thm:weights}.}
The upper bound in the general case follows from Lemmas~\ref{thm:reduction} and~\ref{thm:SubsetSearchUpper}. 
For the lower bound (the general case), consider an instance $I$ in which all targets have weight $w>0$, 
and $W=tw$, for some $t\in {\mathbb R}^+$. Then any strategy must locate $\lceil t \rceil$ targets; from~\cite{multi-target}, the competitive
ratio is at least $\phi(m-t)=\phi(m-\lceil \frac{W}{w_{\max,I}}\rceil)$. The upper and lower bounds in the extreme case are given by
Lemma~\ref{thm:spastiko.case}.
\end{proof}
\bigskip

\proof{Proof of Theorem~\ref{thm:subsets}.}
The upper bound in the general case follows from Lemmas~\ref{thm:reduction} and~\ref{thm:SubsetSearchUpper}. 
For the lower bound in the general case,
given $s<m$, consider an instance $I$ in which $W=sw$, for some $w>0$, and there are $s$ targets of weight $w$, 
while all other targets have weight $0$. Suppose that $s-1$ of positive-weight 
targets are very close to the origin, and can be found at negligible cost. Thus, on this instance, weighted search is as hard
as locating one target in $m-(s-1)$ rays, and thus has competitive ratio 
\[ 1+2\frac{b_{m-s+1,1}^{m-s+1}}{b_{m-s+1,1}-1}=
1+2\frac{b_{m-s}^{m-s+1}}{b_{m-s}-1}=\phi(m-s),
\]
where the second equality follows from~\eqref{eq:defb}. 
The upper and lower bounds in the extreme case are given by
Lemma~\ref{thm:spastiko.case}.
\endproof
\bigskip

\section{The competitive ratio of \AdSub}
\label{sec:analysis}
In this section we determine the competitive ratio of \AdSub, by proving Lemma~\ref{thm:SubsetSearchUpper}. 
The following two lemmas, that are necessary but quite technical, 
provide some useful properties of the functions $\phi$ and $b$. 
\begin{lemma}
\label{lem:boundsf}
The function $\phi$ is increasing and
$
	\phi(q) - \phi(q-1) \ge 2e,
$
for all $q \in \mathbb{N}^+$, where we define $\phi(0) = \lim_{x \to 0} \phi(x) = 3$. 
\end{lemma}

\medskip
\proof{Proof.} 
Define the function
\[
	h(x) = (x+1)\left(1 + \frac{1}{x}\right)^x.
\]
Then $\phi(x) = 1 + 2h(x)$, and the monotonicity of both $\phi$ and $h$ follows from the simple fact that $(1 + 1/x)^x$ is increasing.  Moreover, by direct computation we obtain that $\phi(1) - \phi(0) > 2e$ and $\phi(2)-\phi(1) > 2e$. 
Let 
\[
	H(x) = h(x) - h(x-1).
\]
We will argue that
\begin{equation}
\label{eq:ratioH}
	H(x) \le H(x-1)  \quad\text{for}\quad x\ge 3.
\end{equation}
Given~\eqref{eq:ratioH} the claim that $\phi(q) - \phi(q-1) \ge 2e$, for all $q \in \mathbb{N}^+$, can be shown as follows. Since $(1 + 1/x)^x$ is monotone increasing, we have that
\begin{align*}
	\lim_{x \to \infty}  H(x) &= \lim_{x \to \infty} (x+1)\left(1 + \frac1x\right)^x - x\left(1 + \frac1{x-1}\right)^{x-1} \\
	& \le \lim_{x \to \infty} x\left(1 + \frac1{x-1}\right)^{x-1} - (x-1)\left(1 + \frac1{x-2}\right)^{x-2} 
	\tag{From~\eqref{eq:ratioH}}\\
	& \leq \lim_{x \to \infty} \left(1 + \frac1{x-2}\right)^{x-2} 
    \tag{monotonicity of $(1 + 1/x)^x$} \\ 
	&= e,
\end{align*}
and 
\[
	\lim_{x \to \infty} H(x) \geq 
	\lim_{x \to \infty} \left(1 + \frac{1}{x-1}\right)^{x-1}= e.
\]
Thus, $\lim_{x \to \infty} H(x) = e$. Moreover,~\eqref{eq:ratioH} guarantees that $(H(q))_{q \in \mathbb{N}^+}$ is a non-increasing sequence. Summarizing, we obtain that $H(x) \ge e$ for all $x \ge 3$ and thus
\[
	h(x) - h(x-1) \ge e \implies \phi(x) - \phi(x-1) \ge 2e,
\]
as claimed.

In order to show~\eqref{eq:ratioH}, note that it is equivalent to
\[	
	h(x-1) \ge \frac{h(x) + h(x-2)}2 \quad\text{for}\quad x\ge 3.
\]
We will argue that $h''(x) < 0$ for $x \ge 3$, which shows that $h$ is concave for $x \ge 3$ and thus satisfies the above inequality, 
which completes the proof. Basic calculus reveals that
\[
	h''(x) = \frac{h(x) t(x)}{x(x+1)}, \text{ where } t(x) = x(x+1)\log\left(1 + \frac{1}{x}\right)^2 - 1.
\]
Since $h(x)$ and $x(x+1)$ are positive, it remains to show that $t(x) < 0$ for $x \ge 3$. Using the Taylor series expansion of the logarithm we infer that
\[
	\ln(1+y) \le y - y^2/2 + y^3/3 - y^4/4  + y^5/5, \quad |y|<1.
\]
Thus
\[
	t(x) \le x(x+1)\left(\frac1x - \frac1{2x^2} +\frac1{3x^3}-\frac1{4x^4}+\frac1{5x^5}   \right)^2 - 1.
\]
By expanding the right-hand side we obtain that 
\[
	t(x)\leq \frac {\alpha_7{x}^{7}+\alpha_6{x}^{6}+\alpha_5{x}^5+
\alpha_4{x}^4+\alpha_3{x}^3+\alpha_2{x}^2+\alpha_1x+\alpha_0}{3600x^9},
\]
with
\[
	\alpha_7 = -300, \alpha_6 = 300, \alpha_5 = -260,\alpha_4 = 1420,
\]
and
\[
	\alpha_3 = -615, \alpha_2 = 345, \alpha_1 = -216, \alpha_0 = 144.  
\]
Note that the upper bound for $t$ is less than 0 for $x\in\{3,4,5\}$. Moreover, the values of the $\alpha_i$'s guarantee that for all $x \ge 3$
\[
	\alpha_{i}x^i + \alpha_{i-1}x^{i-1} 
	= (\alpha_{i}x + \alpha_{i-1})x^{i-1}
	\le 0, \quad i\in\{7,3,1\},
\]
and thus $t(x) \le (-260x^5 + 1420x^4)/3600x^9$. However, for $x\ge 6$, this is negative (since $6 \cdot 260 > 1420$), and the proof is completed.
\endproof
\bigskip
We also obtain the following (crude) consequence of Lemma~\ref{lem:boundsf}.

\begin{corollary}
For every $x,y \in \mathbb{N}^+$, it holds that 
$\frac{b_x^{x+1}}{b_x-1}+2y \leq \frac{b_x^{x+y+1}}{b_{x+y}-1}$.
\label{lemma:keeponlyG}
\end{corollary}
\medskip
\proof{Proof.}
From Lemma~\ref{lem:boundsf} we know that for every $x \in \mathbb{N}^+$, $\frac{b_{x}^{x+1}}{b_{x}-1} - \frac{b_{x-1}^{x}}{b_{x-1}-1} \ge e > 2$. 
Using a simple inductive argument, it follows that for every $y \in \mathbb{N}^+$, 
\[
\frac{b_x^{x+y+1}}{b_{x+y}-1} -\frac{b_x^{x+1}}{b_x-1} \geq e\cdot y>2y.
\]
\endproof
\bigskip
We will also make use of the following rather technical estimate.
\begin{lemma}
\label{lem:technicality}
For $q,l \in \mathbb{N}$ let
\[
	h_{q,\ell} = b_{q+1}^{-\ell-1} \left(\frac{b_q^\ell}{b_q-1} + 1 + b_{q+1}\right) (b_{q+1} - 1).
\]
Then $h_{q,l} \leq 1$ for all $1 \le \ell \le q+1$.
\end{lemma}
\medskip
\proof{Proof.}
Substituting the value of $b_q$ and $b_{q+1}$ yields after some straightforward algebraic manipulations
\begin{equation}
\label{eq:hqelltmp}
	h_{q,\ell} = \left(q\Big(1 + \frac1{q(q+2)}\Big)^\ell +  \frac{2q+3}{q+1}\Big(1 - \frac1{q+2}\Big)^{\ell}\right)\frac1{q+2}.
\end{equation}
When viewed as a function of $\ell$, $h_{q,\ell}$ is of the form $aX^\ell + bY^\ell$,  for some $a,b,X,Y > 0$. The second derivative of this function is $aX^\ell \ln(X)^2 + b Y^\ell \ln(Y)^2$. The positivity of $a,b,X,Y$ implies that this is always non-negative, and thus $h_{q,\ell}$ is convex in $\ell$; we obtain that $h_{q,\ell} \le \max\{h_{q, 1}, h_{q,q+1}\}$ for all $\ell \in [1,q+1]$.

It is straightforward to verify that $h_{q,1} = 1$. In the remainder we show that $h_{q,q+1} \le 1$, which will complete the proof. The cases $q = 1,2,3$ are also verified immediately, so we can further assume  that $q \ge 4$. To bound $h_{q,q+1}$ we will first prove the auxiliary fact 
\begin{equation}
\label{eq:auxUpperBound}
	(1 + y)^N \le 1 + yN + y^2N^2/2, ~~ \text{for all } N \in \mathbb{N}_0 \text{ and } y \le N^{-2}.
\end{equation}
To show~\eqref{eq:auxUpperBound}, let us fix $N$ and $y$ as required. We will show by induction that $(1 + y)^n \le  1 + yn + y^2n^2/2$ for all $0 \le n \le N$. The case $n = 0$ is immediate. Moreover, for $0 \le n \le N-1$, by applying the induction hypothesis, we have 
\[
\begin{split}
	(1 + y)^{n+1} & \le (1 + y)(1 + yn + y^2n^2/2)
	= 1 + y(n+1) + (n^2 + 2n + yn^2)y^2/2.
\end{split}
\]
Together with $y \le N^{-2}$ and $n \le N$ this establishes~\eqref{eq:auxUpperBound}.

Before we proceed with bounding $h_{q,q+1}$ we make two auxiliary observations. First, using the bound in~\eqref{eq:auxUpperBound} and that obviously $q(q+2) \ge q^2$ we readily obtain
\[
	\Big(1 + \frac1{q(q+2)}\Big)^{q+1}
	\le \Big(1 + \frac1{q(q+2)}\Big)  \Big(1 + \frac1{q+2} + \frac1{2(q+2)^2}\Big).
\]
Second, using that $1 - x \le e^{-x}$, we obtain
\[
	\Big(1 - \frac1{q+2}\Big)^{q+1}
	= 	\Big(1 - \frac1{q+2}\Big)^{q+2} \cdot  \Big(1 + \frac1{q+1}\Big)^{-1}
	\le \frac1e \cdot \Big(1 + \frac1{q+1}\Big).
\]
By substituting both auxiliary observations in~\eqref{eq:hqelltmp} and collecting terms we obtain that there are polynomials $P,Q$ such that
\[
	h_{q,q+1} \le \left(q\Big(1+\frac1{q(q+2)}\Big)\Big(1 + \frac1{q+2}+\frac1{2(q+2)^2}\Big) + \frac{2q+3}{(q+1)e}\Big(1+\frac1{q+1}\Big)\right)\frac1{q+2}= 1 - P(q)/Q(q),
\]
where $Q(q) = 2e(q+2)^4(q+1)^2 > 0$ and $P(q) = \sum_{i=0}^5 \alpha_i q^i$ is a polynomial of degree 5 such that
\[
	\alpha_5 = 2e-4, \alpha_4 = 17e-38, \alpha_3 = 56e-144,
\]
all $> 0$, 
and
\[
	\alpha_2 = 88e-272, \alpha_1 = 66e-256, \alpha_0 = 19e-96,
\]
all $< 0$. Since $\alpha_5 > 0$ the polynomial $P$ is eventually positive, implying that $h_{q,q+1} \le 1$ whenever $q$ is sufficiently large. Moreover, let us write
\[
	P(q) = q^2(q^3\alpha_5 + \alpha_2) + q(q^3\alpha_4 + \alpha_1) + (q^3\alpha_3 + \alpha_0),
\]
Note that (with plenty of room to spare)
\[
	4^3 \alpha_i > -\alpha_{i-3}, \quad \text{for} \quad i\in\{3,4,5\}.
\]
and thus $q^3 \alpha_i + \alpha_{i-3} > 0$ for all $q\ge 4$ and $i\in\{3,4,5\}$. Thus $P(q) > 0$ for $q \ge 4$ and the proof is completed.
\endproof
\bigskip

We fix some notation that we will use throughout the proof of Lemma~\ref{thm:SubsetSearchUpper}. Let $s = |S|$ denote the cardinality of the sought set of targets, and let $t$ denote the total number of targets discovered by  
{\AdSub} at the moment in which all targets in $S$ have been found. That is, $f = t+1$ upon termination. Clearly, we have
$
	s \le t \le m.
$
For each $1 \le j \le t$ we say that the strategy in in \emph{Phase}~$j$ if the number of targets discovered at that point in time is equal to $j-1$. That is, the 
strategy starts in Phase 1, repeats lines 4--10 in the statement of Algorithm~\ref{alg:exponential.deterministic} until eventually a target is found, proceeds to Phase 2, repeats lines 4--10 in the statement until the second target is found, and so forth. Let $\ell_j \ge 1$ be the number of iterations of the loop in lines 4--10 performed in the $j$-th phase, with $1 \le j \le t$. Note that in Phase $j$, the number of rays explored unsuccessfully (i.e., without finding a target)
equals $\ell_j-1$, and in the $\ell_j$-th iteration a target was discovered. Moreover, let $D_j$ denote the distance of the $j$-th discovered target, $1 \le j \le t$. Since the last discovered target must be in $S$, we obtain the straightforward bound
\begin{equation}
\label{eq:OPTlb}
	d_S \ge D_t.
\end{equation}
A short roadmap to the proof of Lemma~\ref{thm:SubsetSearchUpper} is as follows. We will assume first the more general case $|S| < m$; at the end we argue how to handle the special case $|S| = m$. In the first step, we derive an upper bound on the algorithm's cost; see Lemma~\ref{lemma:main.upper}. This is a function of several parameters, in particular the distances of all the targets that have been discovered, and the cost of the unsuccessful explorations made by the searcher. In the second step we derive upper and lower bounds on the distances of the discovered targets; this is accomplished in Lemma~\ref{lemma:main.lower}. Note that the lower bound is particularly useful, since the last discovered target has to be in $S$ and provides a lower bound for $d_S$. In the third (and most technical) step, we combine all these bounds, using an inductive argument, to arrive at the desired conclusion. This will be accomplished in Lemma~\ref{lemma:G}.

We begin by deriving a bound on the total distance traversed by the searcher, and we write $c(I) = \text{c({\AdSub}, $I$)}$ for brevity. 
Define $Y_j = \prod_{1 \le i < j} b_{m,i}^{\ell_i - 1}$.
\begin{lemma}
For the search cost $c(I)$ incurred by the strategy on instance $I$ 
\begin{equation}
\label{eq:upperSubset}
	c(I) \le   \frac{2Y_{t+1}b_{m,t}}{b_{m,t}-1}
	+ 2 \sum_{1 \le j < t} (D_j + Y_{j+1}) + D_t.
\end{equation}
\label{lemma:main.upper}
\end{lemma}

\medskip
\proof{Proof.}
Note that in Phase 1 the strategy explores in a cyclic fashion the rays to distances $b_{m,1}, b_{m,1}^2, \dots, b_{m,1}^{\ell_1 - 1}$ (going through each ray twice), and then discovers a target at distance $D_1$. Similarly, in Phase 2 it explores the rays to distances $b_{m,1}^{\ell_1 - 1} b_{m,2}, \dots, b_{m,1}^{\ell_1 - 1}b_{m,2}^{\ell_2 - 1}$ and then discovers a target at distance $D_2$.
More generally, in Phase $j$, $1\le j < t$, the distance traversed equals
\[
	2 \cdot \prod_{1 \le i < j} b_{m,i}^{\ell_i - 1} \cdot \sum_{1 \le i < \ell_j} b_{m, j}^i + 2D_j 
	=  2 \cdot \prod_{1 \le i < j} b_{m,i}^{\ell_i - 1} \cdot \frac{b_{m,j}^{\ell_j}-b_{m,j}}{b_{m,j} - 1} + 2D_j.
\]
(As usual, the empty product equals 1.)
Similarly, and because after locating the $t$-th target the searcher does not return to the origin, the distance traversed in the final phase equals
\[
	2 \cdot \prod_{1 \le i < t} b_{m,i}^{\ell_i - 1} \cdot \frac{b_{m,t}^{\ell_t}-b_{m,t}}{b_{m,t} - 1} + D_t.
\]
Using the fact that $Y_j = \prod_{1 \le i < j} b_{m,i}^{\ell_i - 1}$, where $1 \le j \le t+1$, we obtain
\begin{equation}
	c(I) = 2 \sum_{1 \le j \le t} Y_j \frac{b_{m,j}^{\ell_j}-b_{m,j}}{b_{m,j} - 1}
	+ 2 \sum_{1 \le j < t} D_j + D_t.
\label{eq:cost.upper.split1}
\end{equation}
Note that for $1 \le j \le t$
\[
	Y_j \frac{b_{m,j}^{\ell_j}-b_{m,j}}{b_{m,j} - 1}
	 = (Y_{j+1} - Y_j)\frac{b_{m,j}}{b_{m,j}-1}
	 = (Y_{j+1} - Y_j)(m-j+1), 
\]
and consequently
\begin{eqnarray*}
\sum_{1 \le j \le t} Y_j \frac{b_{m,j}^{\ell_j}-b_{m,j}}{b_{m,j} - 1} &=& \sum_{1 \le j \le t} (Y_{j+1} - Y_j)(m-j+1) \nonumber \\
&=& Y_{t+1} \frac{b_{m,t}}{b_{m,t}-1}+ \sum_{1\leq j <t} Y_{j+1}-Y_1 \cdot m \leq Y_{t+1} \frac{b_{m,t}}{b_{m,t}-1}+ \sum_{1\leq j <t} Y_{j+1}.
\label{eq:cost.upper.split2}
\end{eqnarray*}
Substituting this into~\eqref{eq:cost.upper.split1} yields the lemma.
\endproof
\bigskip

Having accomplished the task of providing an appropriate upper bound for {\AdSub} we proceed with deriving explicit bounds for the distances of the targets located by the strategy. 
\begin{lemma}
For the distance $D_j$ of the $j$-th discovered target by {\AdSub} it holds that 
\begin{equation}
\label{eq:upperDj}
	D_j \le Y_j b_{m,j}^{\ell_j} = Y_{j+1} b_{m,j}, \quad \textrm{for all } \ 1 \le j \le t.
\end{equation}
Moreover, suppose that the $j$-th discovered target is at ray $0 \le r \le m-j$, and let $j'$ with $1 \le j' \leq j$ be such that ray $r$ was last explored in Phase $j'$, prior to discovering the target in Phase $j$.
Then
\begin{equation}
\label{eq:lowerDj}
	D_j
	\ge
	Y_{j'} \cdot b_{m,j'}^{-m + j + \sum_{j' \le i \le j} (\ell_i-1)}.
\end{equation}
\label{lemma:main.lower}
\end{lemma}

\proof{Proof.}
In Phase $j$, with $1 \le j \le t$, the largest distance to which a ray is explored is $Y_j \cdot b_{m,j}^{\ell_j}$. Thus
$D_j \le Y_j b_{m,j}^{\ell_j} = Y_{j+1} b_{m,j}$.
For the lower bound, note that there exists $x_j$ with $1 \le x_j \le \ell_{j'}-1$ such that this ray was explored up to a depth of $Y_{j'} b_{m,j'}^{x_j}$, and so $D_j \ge Y_{j'} b_{m,j'}^{x_j}$. 
The number of targets discovered between the last time ray $r$ was explored before discovering the target at distance $D_j$, and the time at which this target was discovered is equal to  $j-j'+1$ (including the said target). Since the number of remaining rays in Phase $j'$ equals $m-j'+1$, we obtain
\[
	m - j' + 1 = (j - j'+1) + \sum_{i = j'+1}^j (\ell_i - 1) + (\ell_{j'}-1-x_j).
\]
Thus,
$
	\sum_{j' \le i \le j} (\ell_i - 1) -m + j = x_j,
$
and so
$
	D_j
	\ge
	Y_{j'} \cdot b_{m,j'}^{-m + j + \sum_{j' \le i \le j} (\ell_i-1)}.
$
\endproof
\bigskip

Moreover, for $x_j$ as defined in the proof of Lemma~\ref{lemma:main.lower}, since $1 \le x_j \le \ell_{j'}-1$ we obtain the bounds
\begin{equation}
\label{eq:boundsL}
	\sum_{j' \le i \le j}(\ell_i-1) \ge m-j+1
	\text{ and }
	\sum_{j' < i \le j}(\ell_i-1) \le m-j,
\end{equation}
that will be useful later. We also obtain a simpler lower bound on $D_j$, namely
\begin{equation}
\label{eq:lowerDjweak}
	D_j \ge Y_j \cdot b_{m,j}^{\ell_j - m + j -1} \quad \text{for all} \quad \ell_j \ge 1.
\end{equation}
To see why~\eqref{eq:lowerDjweak} holds, note that from definition of $Y_j$, we have that
\[
Y_{j'}=Y_j \prod_{j'\leq i< j} b_{m,i}^{1-\ell_i} \geq Y_j \prod_{j'\leq i< j} b_{m,j}^{1-\ell_i}= Y_j b_{m,j}^{\ell_j-1} \prod_{j'\leq i \leq j} b_{m,j}^{1-\ell_i} \geq Y_j b_{m,j}^{\ell_j-1} b_{m,j}^{j-m},
\] 
where the first inequality follows from the fact that $b_{m,h}$ is increasing in $h$, and the last inequality follows from~\eqref{eq:boundsL}.

We introduce some additional notation to facilitate the further analysis. Combining~\eqref{eq:lowerDjweak} with~\eqref{eq:upperDj} we infer that for every $1 \le j \le t$ there exist $\gamma_j$ such that
\[
	D_j = \gamma_j \cdot Y_{j+1}, \quad \text{where} \quad  b_{m,j}^{-m+j} \le \gamma_j \le b_{m,j}.
\]
Moreover, let $J$ be the set of indexes in $\{0, \dots, t-1\}$ such that for each $j \in J$ we have that the target  discovered in Phase $j$ is in $S$. That is, we have $|J| = |S|-1$ and moreover, strengthening the bound in~\eqref{eq:OPTlb}
\begin{equation}
\label{eq:OPTlb2}
	d_S \ge \sum_{j \in J} D_j + D_t = \sum_{j \in J} \gamma_j \cdot Y_{j+1} + D_t.
\end{equation}
Let also $\overline{J} = \{1, \dots, t - 1\} \setminus J$.

With this notation in place, we can now bound the competitive ratio of {\AdSub}, towards the proof of Lemma~\ref{thm:SubsetSearchUpper}. 
Specifically, by 
applying Lemma~\ref{lemma:main.upper} and~\eqref{eq:OPTlb2}, we have 
\begin{eqnarray}
	 \frac{c(I)}{d_S}
	&\le& \frac{\frac{2Y_{t+1}b_{m,t}}{b_{m,t}-1}
	+ 2 \sum_{1 \le j < t} (1 + \gamma_j)Y_{j+1} + D_t}{\sum_{j \in J} \gamma_j \cdot Y_{j+1} + D_t} \nonumber \\
	&=& \frac{\frac{2Y_{t+1}b_{m,t}}{b_{m,t}-1}
	+ 2 (\sum_{j \in J} (1 + \gamma_j)Y_{j+1} +\sum_{j \in \overline{J}} (1 + \gamma_j)Y_{j+1}) + D_t}{\sum_{j \in J} \gamma_j \cdot Y_{j+1} + D_t}
	\nonumber \\ 
	 &\le&   2\max\Big\{\max_{j \in J}\frac{1 + \gamma_j}{\gamma_j}, H\Big\} + 1,
\label{eq:cIdS}
\end{eqnarray}
where $H$ is defined as 
\begin{equation}
\label{eq:H}
	H = \frac{\frac{Y_{t+1}b_{m,t}}{b_{m,t}-1} + \sum_{j \in \overline{J}}Y_{j+1}(1 + \gamma_j)}{D_t}.
\end{equation}
Here, we used the fact that for $a_1, \dots, a_N, b_1, \dots, b_N > 0$, we have 
\[
\frac{\sum_{i=1}^N a_i}{\sum_{i=1}^N b_i} \le \max_{1 \le i\le n}\left\{\frac{a_i}{b_i}\right\}.
\]
Note that the function $(1+x)/x$ is decreasing in $x$. Therefore, as $\gamma_j \ge b_{m,j}^{-m+j}$, for any $j \in J$ we have 
\begin{equation}
\frac{1 + \gamma_j}{\gamma_j} \le 1 + b_{m,j}^{m-j}=1+\left(1+\frac{1}{m-j}\right)^{m-j}\leq 1+e< 4.
\label{eq:gamma_j}
\end{equation}
Therefore,~\eqref{eq:cIdS} gives
\begin{equation}
\label{eq:Hremains}
	{c(I)}
	\le \max\left\{9, 1 + 2H\right\} \, d_S.
\end{equation}
Recall that we assume, for the general case that $|S|<m$; at the end we will argue how to handle the extreme case $|S|=m$. The crucial step in the proof 
of Lemma~\ref{thm:SubsetSearchUpper} will be to show that 
\begin{equation}
H \le \frac{b_q^{q+1}}{b_q - 1}, \ \textrm{where } q = m - |S|.
\label{eq:Hmain}
\end{equation}
This will suffice to prove the lemma in the case $|S|<m$. This is because, from Lemma~\ref{lem:boundsf}, it follows that $1+2H\geq 9$, therefore from~\eqref{eq:Hremains}, the competitive ratio of \AdSub \ is at most $1+2H\leq\phi(m-|S|)$, for $|S|<m$. 
Note that~\eqref{eq:Hremains} confirms what one expects intuitively, namely that the contribution from targets in $S$ to the competitive ratio is not significant, and bounded by 9, whereas the contribution of targets that are not in $S$ to the competitive ratio is more substantial, and more challenging to bound. In order to bound $H$,  
we first define $L_{a,b}$ and $Y_{a,b}$ as 
\[
L_{a,b} = \sum_{a \le i \le b} (\ell_i-1)
	~\text{ and }~
	Y_{a,b} = \frac{Y_{b+1}}{Y_{a}}.
\]
Using the definition of $H$~\eqref{eq:H} , the lower bound on $D_t$~\eqref{eq:lowerDj} and the fact $\gamma_j \le b_{m,j}$,
we obtain that for some $t' \le t$,
\begin{equation}
	H \le b_{m,t'}^{m-t-L_{t',t}}\left(\frac{b_{m,t}Y_{t', t}}{b_{m,t}-1} + \sum_{j \in \overline{J}}(1 + b_{m,j})Y_{t',j} \right).
\label{eq:Hfirst}
\end{equation}
This expression is central in our analysis and will indeed help us bound the contribution of targets that are not in $S$.
We first define a partition of ${\overline J}$ into two sets which contain elements in $\overline{J}$ that are greater than,
and respectively smaller than or equal to $t-1$;
namely we define ${\overline J}_> = {\overline J} \cap \{t', \dots, t-1\}$ 
and ${\overline J}_\le = {\overline J} \cap \{1, \dots ,t'-1\}$.
Regarding any $j \in {\overline J}_\le$, note that $Y_{t',j} \le 1$ and hence
$
	b_{m,t'}^{m-t-L_{t',t}} (1 + b_{m,j})Y_{t',j}
	\le 	b_{m,t'}^{m-t-L_{t',t}} \cdot (1 + b_{m,j}).
$
From~\eqref{eq:boundsL} we infer that $L_{t',t} \ge m-t+1$ and so the previous expression is bounded by at most $(1 + b_{m,j}) / b_{m,t'}$. Thus, we obtain that
\begin{equation}
\label{eq:boundjsmall}
	b_{m,t'}^{m-t-L_{t',t}} (1 + b_{m,j})Y_{t',j}
	< \frac{1+b_{m,t'}}{b_{m,t'}} \le 2, \quad j \in {\overline J}_\le.
\end{equation}
In words, \eqref{eq:boundjsmall} shows that each element in ${\overline J}_\le$ contributes at most an additive 2 to the bound in~\eqref{eq:Hfirst}. Moreover, note that for every $j \in {\overline J}_\le$ 
\[
(1+b_{m,j})Y_{t',j} \leq 1+b_{m,j}=1+\frac{1}{m-j}<1+\frac{1}{m-t'},
\]
whereas for every $j \in {\overline J}_>$ 
\[
(1+b_{m,j})Y_{t',j} \geq 1+b_{m,j}=1+\frac{1}{m-j}\geq 1+\frac{1}{m-t'}.
\]
Therefore, the contribution to $H$ of every $j \in {\overline J}_\le$ is smaller than the corresponding contribution of every $j \in {\overline J}_>$. Thus, in order to bound $H$ we can assume, without loss of generality, that ${\overline J}_>$ has maximal cardinality. 

Combining~\eqref{eq:Hfirst},~\eqref{eq:boundjsmall} and the observation on the maximality of ${\overline J}_>$, we have that 
\begin{equation}
H \leq G+2|{\overline J}_\le|,
\quad \textrm{where} \quad G=b_{m,t'}^{m-t-L_{t',t}}\left(\frac{b_{m,t}Y_{t', t}}{b_{m,t}-1} + \sum_{j \in \overline{J}_>}(1 + b_{m,j})Y_{t',j} \right).
\label{eq:HGJ} 
\end{equation}
The following is the main technical lemma of this section. The lemma will help us bound $G$, which in turn will help us prove~\eqref{eq:Hmain}.
\begin{lemma}
\[
G \leq \frac{b^{Q+|\overline{J}_>|+1}_{Q+|\overline{J}_>|}}{b_{Q+|\overline{J}_>|}-1}, \textrm{ where } Q=m-t. 
\]
\label{lemma:G}
\end{lemma}
\proof{Proof.}
For ease of notation, let us denote $|\overline{J}_>|$ by $d$. Since ${\overline J}_> = {\overline J} \cap \{t', \dots, t-1\}$, and since we can assume that ${\overline J}_>$ has maximal cardinality, we have that $d=t-t'$. 

We will prove the lemma by induction on $d$. The base case, $d=0$, follows trivially by direct substitution. To give some intuition into the inductive step, let us also show the lemma in the case $d=1$, namely when $t'=t-1$. By substituting into the expression of $G$~\eqref{eq:HGJ}, we have that 
\[
G= b_{Q+1}^{Q-(\ell_t-1)}\left(\frac{b_{Q}^{\ell_t}}{b_{Q}-1} + (1 + b_{Q+1})\right) \leq b_{Q+1}^{Q-(\ell_t-1)} \frac{b^{\ell_t+1}_{Q+1}}{b_{Q+1}-1}=
\frac{b^{Q+2}_{Q+1}}{b_{Q+1}-1},
\]
where the inequality follows from Lemma~\ref{lem:technicality}, since $\ell_t \ge 1$.

For the inductive step, let us denote by $G_d$ the value of $G$ given that $t-t'=d$. By substituting into~\eqref{eq:HGJ}, we obtain that $G_d$ can be expressed as 
\begin{equation}
G_d=b_{Q+d}^{Q-\sum_{x=0}^d(\ell_{t-x}-1)}
\left(
\frac{b_Q \prod_{x=0}^d b_{Q+x}^{\ell_{t-x}-1}}{b_Q-1}
+ \sum_{y=t-d}^{t-1} (1+b_{m-y}) \prod_{x=t'}^y b_{m,x}^{\ell_x-1}
\right).
\label{eq:G_d}
\end{equation}
We also have that
\begin{equation}
\prod_{x=0}^d b_{Q+x}^{\ell_{t-x}-1}=b_{Q+d}^{\ell_{t-d}-1} \prod_{x=0}^{d-1} b_{Q+x}^{\ell_{t-x}-1}, 
\label{eq:b_{Q+x}}
\end{equation}
and
\begin{eqnarray}
\sum_{y=t-d}^{t-1} (1+b_{m-y}) \prod_{x=t-d}^y b_{m,x}^{\ell_x-1} &=&
\sum_{y=t-(d-1)}^{t-1} (1+b_{m-y}) \prod_{x=t-d}^y b_{m,x}^{\ell_x-1}+ (1+b_{Q+d}) b_{Q+d}^{\ell_{t-d}-1} \nonumber \\
&=& b_{Q+d}^{\ell_{t-d}-1} \left( 
\sum_{y=t-(d-1)}^{t-1} (1+b_{m-y}) \prod_{x=t-(d-1)}^y b_{m,x}^{\ell_x-1}+ (1+b_{Q+d})
\right).
\label{eq:1+b_{m-y}}
\end{eqnarray}
Denote by $F_d$ the quantity $b_{Q+d}^{Q-\sum_{x=0}^d(\ell_{t-x}-1)}$. By substituting~\eqref{eq:b_{Q+x}} and~\eqref{eq:1+b_{m-y}} 
into~\eqref{eq:G_d} we have that
\[
G_d=F_d \cdot b_{Q+d}^{\ell_{t-d}-1} \left(
\sum_{y=t-(d-1)}^{t-1} (1+b_{m-y}) \prod_{x=t-(d-1)}^y b_{m,x}^{\ell_x-1}+ (1+b_{Q+d})
\right),
\]
and we can thus 
obtain a recursive expression for $G_d$, namely
\begin{equation}
G_d=F_d \cdot b_{Q+d}^{\ell_{t-d}-1} \left(
\frac{G_{d-1}}{F_{d-1}}+1+b_{Q+d}
\right).
\label{eq:recursive_Gd}
\end{equation}
From the induction hypothesis, we have 
\[
G_{d-1} \leq \frac{b_{Q+d-1}^{Q+d}}{b_{Q+d-1}-1}. 
\]
Let $L_i$ denote the expression $\sum_{x=0}^i (\ell_{t-x}-1)$. Then $F_{d-1}=b_{Q+d-1}^{Q-L_{d-1}}$, and $F_d=b_{Q+d}^{Q-L_d}$. Therefore,
$\frac{G_{d-1}}{F_{d-1}}= {b_{Q+d-1}^{L_{d-1}+d}}/{b_{Q+d-1}-1}$, and from~\eqref{eq:recursive_Gd} we obtain
\begin{align}
G_d &\leq b_{Q+d}^{Q-L_d} b_{Q+d}^{\ell_{t-d}-1} 
\left(
\frac{b_{Q+d-1}^{L_{d-1}+d}}{b_{Q+d-1}-1}+1+b_{Q+d}
\right) &\nonumber \\
&\leq 
b_{Q+d}^{Q-L_{d-1}} \frac{b_{Q+d}^{L_{d-1}+d+1}}{b_{Q+d}-1} &  \tag{Lemma~\ref{lem:technicality}} \nonumber \\
&= \frac{b_{Q+d}^{Q+d+1}}{b_{Q+d}-1}, \nonumber
\end{align}
where the conditions for applying Lemma~\ref{lem:technicality} follow from the fact that $L_{d-1} \leq Q$, from~\eqref{eq:lowerDjweak}.
This completes the inductive step and the proof of the lemma.
\endproof
\bigskip

We are now ready to prove the competitiveness of \AdSub \ using the above lemma. 

\bigskip
\proof{Proof of Lemma~\ref{thm:SubsetSearchUpper}.}
Consider first the case $|S|<m$. Recall that it suffices to show~\eqref{eq:Hmain}.
We have $|{\overline J}| = |\{1, \dots, t-1\}\setminus J| = t-1 - (|S|-1) = t-|S|$. Moreover, $|{\overline J}_>|=d$. Therefore, 
$|{\overline J}_\leq| =t-|S|-d$. Since $H\leq G+2|{\overline J}_\leq|$, from Lemma~\ref{lemma:G} we obtain that
\[
H \leq \frac{ b_{m-t+d}^{m-t+d+1}}{b_{m-t+d}-1} + 2(t-|S|-d) \leq \frac{b_{m-|S|}^{m-|S|+1}}{b_{m-|S|}-1},
\]
where the second inequality follows from Corollary~\ref{lemma:keeponlyG}. We conclude that~\eqref{eq:Hmain} holds.

Remains to consider the extreme case $|S|=m$.
Here, we adapt~\eqref{eq:upperSubset} to this case. 
Recall that $J$ is  the set of indexes  in $\{0, \dots, t-1\}$ such that for each $j \in J$ we have that the target  discovered in Phase $j$ is in $S$. Thus $J = \{0, \dots, t-1\}$ and ${\overline J} = \emptyset$ in this case. Moreover, note that $\ell_t = \ell_m = 1$, as the last remaining ray is searched until a target is found without the searcher returning to the origin. Then, the bound in~\eqref{eq:upperSubset} changes to
\[
	c(I) \le \frac{2Y_mb_{m,m-1}}{b_{m,m-1}-1} + 2\sum_{j=1}^{m-2}(D_j + Y_{j+1}) + 2D_{m-1} + D_m.
\]
Note that $b_{m,m-1} = 2$. Proceeding as in~\eqref{eq:cIdS} we get that
\[
	\frac{c(I)}{d_S} \le 1 + 2\max\left\{\max_{j \in J}\frac{1+\gamma_j}{\gamma_j}, \frac{2D_{m-1}}{2/3D_{m-1}}, \frac{4Y_m}{D_{m-1}/3 + D_m} \right\}.
\]
From~\eqref{eq:gamma_j}, we know that $\frac{1+\gamma_j}{\gamma_j} \leq 1+e$. Moreover, we have that $D_m \ge Y_m$ and $D_{m-1} \ge Y_m/2$ 
by~\eqref{eq:lowerDjweak}. This implies that $\frac{c(I)}{d_S} \leq 1+2(1+e)$, and the proof is completed. 
\endproof
\bigskip

\section{Conclusion}
\label{sec:conclusions}

In this work we introduced and studied a generalization of the star search problem in which each ray has a weighted target, and the objective for the searcher is to collect a certain aggregate weight accrued by located targets. We showed that weighted search can be reduced to another search problem, namely the subset search problem, and we proposed and analyzed an efficient search strategy for both problems.

It would be interesting to consider weighted search in other domains, for instance by studying search games in which the searcher and the (weighted) hiders can use mixed strategies. 
Another possible application is in the setting of {\em caching games} such as the scatter hoarder's problem of~\cite{alpern2011search}, in which a hider distributes resources across a number of locations, and a searcher tries to retrieve them given some bound on the total search cost.

\bibliographystyle{plain} 
\bibliography{targets}

\end{document}